\documentclass[10pt,journal]{IEEEtran}

\ifCLASSINFOpdf
\else
	\usepackage[dvips]{graphicx}
\fi
\usepackage{url}
\usepackage[export]{adjustbox}

\usepackage{graphicx}
\usepackage{caption}
\usepackage{subcaption}
\usepackage{balance}
\usepackage{multirow}
\usepackage{booktabs}
\usepackage{tabularx}
\usepackage{array}
\usepackage{algorithm,algpseudocode}

\algnewcommand{\Initialize}[1]{%
	\Statex \textbf{Initialize:}
	\Statex \hspace*{\algorithmicindent}\parbox[t]{.8\linewidth}{\raggedright #1} }

\algnewcommand{\Forward}[1]{%
	\Statex \textbf{Forward pass:}
	\Statex \hspace*{\algorithmicindent}\parbox[t]{.8\linewidth}{\raggedright #1} } \usepackage{amsmath,amssymb}
\usepackage{hyperref}
\usepackage{balance}
\usepackage{csquotes}
\usepackage[inline]{enumitem}
\usepackage[super]{nth}

\usepackage{xcolor}
\newcommand{\comment}[1]{{\footnotesize\color{red}#1}} 
\newcommand{\old}[1]{}

\newcommand{\new}[1]{{\color{black}#1}}
\newcommand{\neww}[1]{{\color{black}#1}}

\DeclareMathOperator{\atantwo}{atan2}

\usepackage{amsmath,amsthm}

\makeatletter
\newcommand{\neutralize}[1]{\expandafter\let\csname c@#1\endcsname\count@}
\makeatother

\newtheorem{lemma}{Lemma}
\newtheorem{corollary}{Corollary}[lemma]
\newtheorem{proposition}{Proposition}

\newenvironment{propositionbis}[1]
  {\renewcommand{\theproposition}{\ref*{#1}$'$}%
   \neutralize{proposition}\phantomsection
   \begin{proposition}}
  {\end{proposition}}

\newtheorem{definition}{Definition}

\begin{document}
\bstctlcite{IEEEexample:BSTcontrol}


\title{Optimizing Multicarrier Multiantenna Systems for LoS Channel Charting}


\author{
	Taha Yassine,
	Luc Le Magoarou,
	Matthieu Crussière,
	Stéphane Paquelet
 \thanks{Taha Yassine and Stéphane Paquelet are with b$<>$com, Rennes, France (email: \{taha.yassine ; stephane.paquelet\}@b-com.com).}
\thanks{Luc Le Magoarou and Matthieu Crussière are with INSA Rennes, CNRS, IETR - UMR 6164, F-35000, Rennes, France (email: \{luc.le-magoarou ; matthieu.crussiere\}@insa-rennes.fr).}
}

\maketitle

\begin{abstract}
	Channel charting (CC) consists in learning a mapping between the space of raw channel observations, made available from pilot-based channel estimation in multicarrier multiantenna system, and a low-dimensional space where close points correspond to channels of user equipments (UEs) close spatially. Among the different methods of learning this mapping, some rely on a distance measure between channel vectors. Such a distance should reliably reflect the local spatial neighborhoods of the UEs. The recently proposed phase-insensitive (PI) distance exhibits good properties in this regards, but suffers from ambiguities due to both its periodic and oscillatory aspects, making users far away from each other appear closer in some cases. In this paper, a thorough theoretical analysis of the said distance and its limitations \new{due to ambiguities} is provided. \new{Consequently, a new channel distance especially designed to remove ambiguities is proposed.} Guidelines for designing systems capable of learning quality charts \new{with the proposed distance} are also derived. Experimental validation is then conducted on synthetic and realistic data in different scenarios.

\end{abstract}


\IEEEpeerreviewmaketitle

\section{Introduction}
\label{sec:introduction}
The next generation of cellular network (6G) standards will most likely be marked by a profound integration and utilization of artificial intelligence and in particular machine learning (ML) methods~\cite{Tang2019, Ali2020, Miltiadis2022, Shi2023, Merluzzi2023}. Indeed, their ability to learn complex patterns from data makes them suitable for many applications where traditional model-based approaches rapidly reach their limit, especially with the advent of massive multiple input multiple output (MIMO) systems where very large bandwidths and a large number of antennas are considered. 

User positioning, which will certainly play an important role in future wireless communication systems, is one such example where ML is expected to bring improvements over traditional approaches. Although dedicated positioning systems like global navigation satellite systems (GNSS) are capable of accurately locating users, they require dedicated hardware, and integration in traditional communication systems is not always feasible nor desired \cite{Studer2018}. They are also known for their bad performance in indoor scenarios, making them not as versatile as needed. Pure radio-based approaches where the existing communication systems are used for user positioning may thus become more attractive. Such approaches are often model-based where they rely on angle-of-arrival (AoA), time-difference-of-arrival (TDoA) and round-trip-time (RTT) measurements~\cite{Dwivedi2021}. These rather basic models work well for simple scenarios, e.g., line-of-sight (LoS) propagation, but struggle when tested on complex environments, e.g., non-line-of-sight (NLoS) propagation. They constitute a great step towards accurate and precise positioning but still fall behind in terms of the performance requirements of the next generations of wireless systems.

The today widespread adoption of massive MIMO systems and large bandwidths means higher resolvability in both the temporal and angular domains. The channel,  being dependent almost exclusively on the transmit position, is then expected to contain enough information to retrieve it. Nonetheless, hand-crafting algorithms capable of this becomes unrealistic as propagation scenarios become more complex. On the other hand, data-driven methods were successfully exploited to learn to locate UEs with great precision~\cite{Wang2015, Vieira2017, Decurninge2018, Ferrand2020}. However, these methods come with their own drawbacks as they usually require a database of labelled data often hard to acquire. This becomes even more challenging as a database needs to be provided for every new environment to train the model. This being said, although many components of a communication system may benefit from the knowledge of UEs' locations, most of them would only require a \emph{relative} positioning of the users.

Recently, channel charting (CC)~\cite{Studer2018,Ferrand2021,Rappaport2021} has emerged as an unsupervised alternative to user positioning, relying solely on channel data to pseudo-locate UEs relatively to each other. It consists in learning a mapping between channel measurements and a low dimensional space, giving rise to a chart. Points on the chart are akin to pseudo-locations in that they preserve the structure of small spatial neighborhoods rather than the absolute locations. Although less informative than absolute locations, they often convey enough information for some location-based applications such as SNR prediction~\cite{Kazemi2020}, beam management~\cite{Ponnada2021}, pilot allocation~\cite{Ribeiro2020} and precoding~\cite{Yassine2022b} among others (see~\cite{Ferrand2023} for a recent survey on the topic). 

Among CC methods, some explicitly rely on a distance measure between channels to build the chart~\cite{Agostini2020, LeMagoarou2021, Altous2022, Stahlke2023}. For them to be successful, the considered channel distance must allow to reliably detect if any two channels measurements were taken close to each other or not in a given area. In other words, the spatial neighbourhoods within the considered area should be \emph{identifiable} by the channel distance measure. One can then raise the following two questions: \emph{What distance is the most adequate for identifiability? And how can a multicarrier MIMO system be calibrated and parameterized to guarantee it?}

\noindent {\bf Contributions.} In this paper, MIMO OFDM systems are considered for the task of CC in LoS conditions. First, the notion of identifiability of an area by a channel distance is rigorously defined and its importance for CC is discussed. Then, the phase-insensitive (PI) distance introduced in~\cite{LeMagoarou2021} and used for CC in~\cite{Yassine2022b,Yassine2022a} is thoroughly analyzed, highlighting its dependency on system parameters such as antenna positions and subcarrier distribution. Conditions to achieve identifiability of an area that link it to these system parameters using the PI distance are given. Special system configurations easing identifiability and consequently particularly suited for CC are proposed. In particular, the following practical guidelines are provided:
\begin{enumerate}
\item A relationship between the radial size of the area to chart and the subcarrier spacing is given in order to avoid long-range ambiguities (Proposition~\ref{prop:nec_ULA}).
\item \new{A new channel distance is proposed, which consists in applying an appropriate threshold on the PI channel distance \cite{LeMagoarou2021}} to avoid short-range ambiguities (Proposition~\ref{prop:suff_ULA}).
\item The advantage of using uniform circular arrays (UCA) instead of ULAs for CC is shown mathematically (Proposition~\ref{prop:nec_UCA}, Proposition~\ref{prop:suff_UCA} and Proposition~\ref{prop:UCA}).
\item A relationship between the radius of the UCA and the bandwidth is given to minimize warping effects on charts (Proposition~\ref{prop:round}).
\end{enumerate}
Finally, experiments are conducted on both synthetic and realistic channels, validating all the theoretical findings of the study.



\noindent {\bf Related work.}
\new{
The initial approaches to CC relied on extracting useful features from channel measurements before effectively learning a chart. In the seminal paper~\cite{Studer2018}, features derived from the raw \nth{2} moment of channel vectors are used as inputs to manifold learning methods. Diverse strategies were then proposed in subsequent work, exploring different aspects of the problem. For example, charting methods can be categorized based on their use or not of deep learning. Work like~\cite{Studer2018,Agostini2020,Deng2018,LeMagoarou2021} rely on traditional dimensionality reduction algorithms such as Sammon's mapping (SM), principal component analysis (PCA), t-distributed stochastic neighbor embedding (t-SNE)~\cite{Maaten2008} and Isomap~\cite{Tenenbaum2000}. On the other hand, \cite{Lei2019,Ferrand2021,Huang2019,Yassine2022a} propose to use deep neural networks of different architecture (e.g. MLPs, CNNs, auto-encoders, etc.) and trained in different ways. Some rely on contrastive learning technics such as Siamese networks and triplet networks as a way to adapt to the unsupervised nature of the task.

Another way of categorizing charting methods is to consider whether or not they rely on an explicit computation of pairwise distances between channel vectors. Indeed, some work exploit the idea that finding a distance measure in the channel space that correlates well with the real spatial distances can allow for learning charts of good quality. While the Euclidean distance is not suitable for channel vectors~\cite{LeMagoarou2021}, finding a good proxy measure can be challenging. Some work relying on a channel distance include:
\begin{itemize}
	\item \cite{Ponnada2021}: the log-Euclidean distance between covariance matrices is used. The computed distance matrix is fed to Isomap and t-SNE to produce charts. 
	\item \cite{Agostini2020}: an Euclidean distance matrix completion (EDMC) approach is presented. It is based on the raw 2$^{\textrm{nd}}$ moment features introduced in~\cite{Studer2018}. A covariance distance measure (CMD) is used to identify close clusters of UEs. A neighborhood graph is constructed from the computation of this distance and only connections under a certain threshold are kept. The edges are however populated with the original Euclidean distance in the feature space and the chart obtained by solving a graph realization problem through an EDMC algorithm.
	
	\item \cite{Stahlke2023}: a multi-anchor time of flight (ToF) based distance is derived and exploited to construct the neighborhood graph. The geodesic distance resulting from it shows great correlation with the spatial Euclidean distance.
	\item \cite{Altous2022}: an earth mover's distance (EMD) applied to super-resolution channel features is introduced and shown to outperform other considered covariance distances (i.e. distances defined on the space of covariance matrices, such as the log-Euclidean distance uses in~\cite{Ponnada2021}).
	\item \cite{LeMagoarou2021}: the proposed PI distance eliminates the effect of the global phase, making it insensitive to small scale fading and more adequate to CC. Isomap is used to produce the final chart.
	\item \cite{Stephan2023}: the angle-delay profile (ADP) distance is proposed. It is inspired by both the PI distance~\cite{LeMagoarou2021} and the ToF-based distance~\cite{Stahlke2023}. It exploits the fact that channel vectors are sparse in the time domain and that most of the power is concentrated around the LoS component.
\end{itemize}
}

This paper is (to the best of the authors' knowledge) the first to propose system configurations explicitly targeted at easing CC. To do so, the PI channel distance is considered~\cite{LeMagoarou2021}, but others could have been also considered.

Although this paper doesn't explore the communication capabilities of the designed system, it could nonetheless be framed as an ISAC problem~\cite{Wymeersch2021}. Indeed, this paper focuses on adapting an existing communication system to CC, which can be seen as a form of sensing. Similarly to~\cite{MateosRamos2022}, the approach presented here could be easily extended to account for a trade-off between communication and CC.

\new{
	\noindent {\bf Notation.} Lowercase boldface letters and uppercase boldface letters denote column vectors and matrices, respectively. For a matrix $\mathbf{A}$, we denote its transpose by $\mathbf{A}^\top$ and its conjugate transpose by $\mathbf{A}^H$. For a vector $\mathbf{a}$, its $l^2$-norm is denoted by $\Vert\mathbf{a}\Vert_2$. A spatial vector of dimension two or three is denoted by $\vec{a}$. The Kronecker product of two matrices $\mathbf{A}$ and $\mathbf{B}$ is denoted by $\mathbf{A}\otimes\mathbf{B}$.
}
\section{Problem formulation}

In this section, the considered setting is introduced, the CC task is briefly presented and the notion of identifiability is properly defined and discussed.

\subsection{System and channel model}

Let us consider a base station (BS) equipped with $N_a$ antennas, located at positions $\vec{p}_1,\dots,\vec{p}_{N_a}$ relative to the barycenter of the antennas, which is considered as the origin. The BS communicates with single antenna UEs over $N_s$ evenly spaced OFDM subcarriers $f_1,\dots,f_{N_s}$. Let $\Delta_f$ denote the subcarrier spacing, $B\triangleq N_s\Delta_f$ denote the total bandwidth, $f_c \triangleq \frac{1}{N_s}\sum_i f_i$ the central frequency and $\lambda \triangleq \frac{c}{f_c}$ the central wavelength.

In order to make the theoretical analysis of the paper clear and interpretable, let us consider a two dimensional (2D) configuration in which the BS antennas and the users are at the same height (all in the same plane) and free space propagation. In that case, the channels have a single LoS path, so that the uplink channel between an user at location $\vec{x}=(r,\theta)$ (in polar coordinates) and the $n$th antenna of the BS on the $m$th subcarrier can be expressed as
\begin{equation}
	h_{mn}(\vec{x}) = h_{mn}(r,\theta) =\frac{1}{r} \mathrm{e}^{-\mathrm{j}\frac{2\pi}{\lambda} \left(r\frac{f_m}{f_c} - \vec{p}_n.\vec{u}(\theta)\right)},
\end{equation}
where $\vec{u}(\theta) = (\cos(\theta),\sin(\theta))$ (in cartesian coordinates) is the unit vector pointing towards the location $\vec{x}$. The complete channel \new{for a given single UE} considering all BS antennas and subcarrier can then conveniently be expressed in vector form as 
\begin{equation}
	\mathbf{h}(\vec{x})=\frac{\sqrt{N_aN_s}}{r}\mathrm{e}^{-\mathrm{j}2\pi \frac{r}{\lambda}}\mathbf{f}(r)\otimes\mathbf{a}(\theta) \in \mathbb{C}^{N_aN_s},
	\label{eq:channel_model}
\end{equation}
where $\mathbf{f}(r)$ and $\mathbf{a}(\theta)$ are the \emph{frequency signature vector} and the \emph{steering vector} respectively, defined as
\begin{equation}
	\mathbf{f}(r)\triangleq\frac{1}{\sqrt{N_s}}\left(\mathrm{e}^{-\mathrm{j}2\pi\frac{r}{c}(f_1-f_c)},\dots,\mathrm{e}^{-\mathrm{j}2\pi\frac{r}{c}(f_{N_s}-f_c)} \right)^\top,
\end{equation}
and
\begin{equation}
	\label{eq:steer_vec}
	\mathbf{a}(\theta)\triangleq\frac{1}{\sqrt{N_a}}(\mathrm{e}^{\mathrm{j}\frac{2\pi}{\lambda}\vec{p}_1.\vec{u}(\theta)},\dots,\mathrm{e}^{\mathrm{j}\frac{2\pi}{\lambda}\vec{p}_{N_a}.\vec{u}(\theta)})^\top.
\end{equation}
The latter is in generic form and encompasses any 2D antenna array configurations. In particular, for a uniform linear array (ULA), $\vec{p}_n$ is given as $(0, (n-1)\Delta_r\lambda)$ where $\Delta_r$ is the normalized antenna separation. Similarly, for a uniform circular array (UCA), although less common, $\vec{p}_n$ is given as $(R\cos(\frac{2\pi(n-1)}{N_a}), R\sin(\frac{2\pi(n-1)}{N_a}))$ where $R$ is the radius of the array.

\subsection{Channel charting}
Let us denote $M = N_aN_s$ the considered channel dimension. CC aims at learning a function $\mathcal{F}$ that maps the channel vectors of UEs in an area of interest $A$ to a low-dimensional space of dimension $D \ll M$,
\begin{align*}
	\mathcal{F}:\; \mathbb{C}^{M} & \to\mathbb{R}^D                                     \\
	\mathbf{h}                & \mapsto\mathcal{F}(\mathbf{h})=\mathbf{\mathbf{z}},
\end{align*}
where $\mathbf{z}$ is the chart location associated to $\mathbf{h}$. CC methods are assessed on their ability to preserve spatial neighborhoods, that is, how well the original spatial distances match the distances on the learned chart for close points
\begin{equation}
	\Vert \vec{x}_i - \vec{x}_j \Vert_2 \approx \Vert \mathbf{z}_i - \mathbf{z}_j \Vert_2.
\end{equation}
The ability to preserve the global geometry can also be considered.
Performance measures such as continuity (CT), trustworthiness (TW) and Kruskal stress (KS) are used to assess the quality of such charts (see~\cite{Lei2019} for a formal definition). They are calibrated using a database of channel measures at different locations
$$\{\mathbf{h}_i\}_{i=1}^{N_c} = \{\mathbf{h}(\vec{x}_i)\}_{i=1}^{N_c},$$ 
$N_c$ being the number of measured channels (training samples), and the two notations can be used indifferently, depending on the wish for brevity or explicitness.

This paper focuses on charting methods based on a distance measure $d : \mathbb{C}^{M} \times \mathbb{C}^{M} \to \mathbb{R}_+$ between pairs of channel vectors (or features derived from it) that tries to match the Euclidean distances between the real locations as closely as possible such that
\begin{equation}
	\Vert \vec{x}_i - \vec{x}_j \Vert_2 \approx d(\mathbf{h}_i,\mathbf{h}_j).
\end{equation}
Such a distance is then used to produce the chart. In general, a distance matrix between all pairs of training channels is computed:
$$
\mathbf{D} \in \mathbb{R}^{N_c \times N_c}, \text{ with } d_{ij} = d(\mathbf{h}_i,\mathbf{h}_j),
$$
and subsequently used as the input to a dimensionality reduction algorithm, such as Isomap or EDMC + ADMM~\cite{Agostini2020}, capable of producing the final chart.

\subsection{Identifiability}

Using an appropriate distance measure is necessary to build good charts. In particular, the distance should allow to detect spatial neighborhoods. In order to make this statement more precise, let us define the identifiability of a set of channel observations by a distance.

\begin{definition}[Strong Identifiability]
	\label{def:strong_iden}
	A set of channel observations and their corresponding locations $S\subset\mathbb{R}^{2}$ are said to be strongly identifiable by the channel distance $d$ iff $\forall \vec{x},\vec{y},\vec{z}\in S$
	\begin{align*}
		d(\mathbf{h}( \vec{y}),\mathbf{h}( \vec{x})) > d(\mathbf{h}( \vec{z}),\mathbf{h}( \vec{x}))
		\Leftrightarrow \lVert\vec{y}-\vec{x}|\rVert_2 > \lVert\vec{z}-\vec{x}|\rVert_2
	\end{align*}
\end{definition}

This definition is described as \emph{strong} because it amounts to finding a channel distance bijective with the Euclidean distance. Since achieving this is  often challenging, a \emph{weak} version of the problem is presented.
\begin{definition}[Weak Identifiability]
	\label{def:weak_iden}

	A set of channel observations and their corresponding locations $S\subset\mathbb{R}^{2}$ is said to be weakly identifiable by the channel distance $d$ iff $\forall \vec{x},\vec{y},\vec{z}\in S$
\begin{align*}
	d(\mathbf{h}( \vec{y}),\mathbf{h}( \vec{x})) > d(\mathbf{h}( \vec{z}),\mathbf{h}( \vec{x}))\\
	\Rightarrow |r_y-r_x| > |r_z-r_x|\ \textrm{or}\ |\theta_y-\theta_x| > |\theta_z-\theta_x|
\end{align*}
\end{definition}
This weak version of identifiability is defined in terms of the polar coordinates to better suit the mathematical model underlying the propagation channels. Obviously, strong identifiability implies weak identifiability, so that weak identifiability is necessary for strong identifiability.

For a CC method based on a distance $d$ to perform well on an area $A$, identifiability is an important pre-requisite. Indeed, it is impossible for the method to preserve spatial neighborhoods (which is its objective) if the channel distance measure it is based on already fails to do so. In the subsequent parts of the paper, the distance measure introduced in~\cite{LeMagoarou2021} is analyzed, leading to the statement of system setting guidelines and a thresholding procedure in order to favor 
 weak identifiability. For two UEs communicating with the BS with corresponding channel vectors $\mathbf{h}_i$ and $\mathbf{h}_j$, this distance is defined as
\begin{equation}
	d^{\star}(\mathbf{h}_i,\mathbf{h}_j)\triangleq\sqrt{2-2s^{\star}(\mathbf{h}_i,\mathbf{h}_j)},
	\label{eq:PI_distance}
\end{equation}
where $s^{\star}(\mathbf{h}_i,\mathbf{h}_j)\triangleq\frac{\vert\mathbf{h}_i^H\mathbf{h}_j\vert}{\Vert\mathbf{h}_i\Vert_2\Vert\mathbf{h}_j\Vert_2}$ is a similarity measure.

The first step of the analysis of the distance is to notice that it can nicely be decomposed into a radial and an angular term. Indeed, when injecting \eqref{eq:channel_model} one can highlight the following decomposition:
\begin{align}
	\begin{split}
		\frac{\vert\mathbf{h}_i^H\mathbf{h}_j\vert}{\Vert\mathbf{h}_i\Vert_2\Vert\mathbf{h}_j\Vert_2}&=\left\vert \left(\mathbf{f}(r_i)\otimes\mathbf{a}(\theta_i)\right)^H\left(\mathbf{f}(r_j)\otimes\mathbf{a}(\theta_j)\right) \right\vert\\
		&\overset{(a)}{=}\left\vert \left(\mathbf{f}(r_i)^H\otimes\mathbf{a}(\theta_i)^H\right)\left(\mathbf{f}(r_j)\otimes\mathbf{a}(\theta_j)\right) \right\vert\\&\overset{(b)}{=}\left\vert \left(\mathbf{f}(r_i)^H\mathbf{f}(r_j)\right)\otimes\left(\mathbf{a}(\theta_i)^H\mathbf{a}(\theta_j)\right) \right\vert\\
		&=\underbrace{\left\vert \left(\mathbf{f}(r_i)^H\mathbf{f}(r_j)\right)\right\vert}_{\bar{f}(r_i,r_j)}\times\underbrace{\left\vert\left(\mathbf{a}(\theta_i)^H\mathbf{a}(\theta_j)\right)\right\vert}_{\bar{a}(\theta_i,\theta_j)},\\
	\end{split}
	\label{eq:s}
\end{align}
where $(a)$ and $(b)$ are because of the transpose property and the mixed-product property of the Kronecker product respectively. This leads to a nice expression of $s^{\star}(\mathbf{h}_i,\mathbf{h}_j)$ as the product of two terms: $\bar{f}(r_i,r_j)$ and $\bar{a}(\theta_i,\theta_j)$. This decomposition in terms of radial and angular component explains why weak identifiability is studied here.

Even if such distance measure is phase insensitive, it is still subject to ambiguities that prohibit the establishment of unique local neighborhoods. Indeed, when moving away from a given reference point in an arbitrary direction, the value of the distance oscillates, meaning that points on that line appear successively farther then closer then farther again with varying amplitudes, creating aliases. In particular, these ambiguities arise from two phenomena: a long-range one and a short-range one. Combined, they create false neighborhoods that affect the trustworthiness of the measure. On one hand, the long-range ambiguities are caused by the periodic nature of the function, meaning that local neighborhoods are duplicated indefinitely in a periodic manner. A simple solution to go around this problem is to only consider users in a area smaller or equal to one period of the function. System parameters could consequently be adjusted to obtain an area of the desired size. One the other hand, the short-range ambiguities are caused by the oscillatory nature of the function inside one period. Fortunately, the amplitudes of the oscillations decrease monotonically, meaning that applying a threshold to only keep the main lobe is feasible and would allow the elimination of the remaining aliases. These considerations are made mathematically rigorous in the next section assuming a ULA for the antenna array.

\section{Limitations of channel charting with a ULA}
\label{sec:ULA}

\new{In this section, the terms of the decomposition \eqref{eq:s} are further developed in the classical case of a base station equipped with a ULA to highlight the limitations of this particular antenna geometry for channel charting.} Note that uniform planar arrays (UPAs) could similarly be analyzed in an extension to the 3D case but are not considered here.

Before proceeding, a review of some classical results~\new{\cite{Proakis2009}} is necessary as proposed in the following Lemmas.

\begin{lemma}
	\label{lem:lemma_1}
	In an OFDM system with evenly spaced subcarriers, the radial term is of the form
	\begin{equation}
		\bar{f}(r_i,r_j)=\left\vert D_{N_s}\left(\frac{2\pi B(r_i-r_j)}{c}\right)\right\vert,
	\end{equation}
	where $D_{N_s}$ is the Dirichlet kernel~\new{\cite{Bruckner1997}} defined as $D_N(x)=\frac{1}{N}\frac{\sin(\frac{x}{2})}{\sin(\frac{x}{2N})}$. It is periodic of period $D_{\bar{f}}=\frac{c}{\Delta_f}$ provided that the absolute delays of both channel observations (i.e., $\tau_i$ and $\tau_j$) are known. Additionally, it contains a main lobe of width $L_{\bar{f}}=\frac{2c}{B}$ centered at $r_i=r_j$.
\end{lemma}
\begin{proof}
	See Appendix~\ref{app:lemma_1}.
\end{proof}


\begin{lemma}
	\label{lem:lemma_2}
	With a ULA, the angular term is of the form
	\begin{equation}
		\label{eq:second_term}
		\bar{a}(\theta_i,\theta_j)=\left\vert D_{N_a}(2\pi\Delta_r N_a (\Theta_i-\Theta_j))\right\vert,
	\end{equation}
	where $\Theta=\sin\theta$. It is of period $D_{\bar{a}}=\frac{1}{\Delta_r}$ w.r.t. $\Theta$. It contains a main lobe centered at $\Theta_i=\Theta_j$ of width $\frac{2}{\Delta_r N_a}$. In addition, it is symmetrical w.r.t. $\theta= \frac{\pi}{2}\;[\pi]$.
\end{lemma}
\begin{proof}
	See Appendix~\ref{app:lemma_2}.
\end{proof}

When taken w.r.t. $\theta$, the $\sin$ in~\eqref{eq:second_term} makes the periodic pattern of varying width in the intervals $[-\frac{\pi}{2},\frac{\pi}{2}]+\pi\mathbb{Z}$.
When the pattern is of width $\pi$, it covers the whole interval, leading to the following corollary:

\begin{corollary}
	\label{cor:cor_21}
	If the ULA antennas are separated by half the wavelength (i.e., $\Delta_r=\frac{1}{2}$), $\bar{a}$ contains a single main lobe $\forall\theta_i,\theta_j\in[-\frac{\pi}{2},\frac{\pi}{2}]$. For a given $\theta_i\in[-\frac{\pi}{2},\frac{\pi}{2}]$, the main lobe extends from $\arcsin(\sin\theta_i-\frac{1}{N_a\Delta_r})$ to $\arcsin(\sin\theta_i+\frac{1}{N_a\Delta_r})$ when $-1+\frac{1}{N_a\Delta_r}\leq\sin(\theta_i)\leq 1-\frac{1}{N_a\Delta_r}$, otherwise it is split in two. For $\theta_i=0$, its width is equal to $L_a=2\arcsin(\frac{1}{N_a\Delta_r})$ (in rad).
\end{corollary}
\begin{proof}
	The proof follows from Lem.~\ref{lem:lemma_2} where the $\arcsin$ is applied to the width of the main lobe to convert it to the angular domain.
\end{proof}

In the case where the main lobe is split in two:
\begin{enumerate}
	\item if $\sin(\theta_i)\leq-1+\frac{1}{N_a\Delta_r}$, one part extends over the span of $[-\frac{\pi}{2},\arcsin(\sin\theta_i+\frac{1}{N_a\Delta_r})]$ and the other over the span of $[\arcsin(\sin(\frac{\pi}{2}-\theta_i)-\frac{1}{N_a\Delta_r}),\frac{\pi}{2}]$;
	\item if $1-\frac{1}{N_a\Delta_r}\leq\sin(\theta_i)$, one part extends over the span of $[-\frac{\pi}{2},\arcsin(\sin(-\frac{\pi}{2}+\theta_i)+\frac{1}{N_a\Delta_r})]$ and the other over the span of $[\arcsin(\sin\theta_i-\frac{1}{N_a\Delta_r}),\frac{\pi}{2}]$.
\end{enumerate}

In the remainder of this paper, ULAs with half a wavelength separated antennas are considered. 

Fig.~\ref{fig:func_plot_ULA} show a 2D color plot of~\eqref{eq:s} along with plots of $\bar{f}$ and $\bar{a}$ for some arbitrary system parameters and for a reference user with associated channel $\mathbf{h}_\textrm{ref}$ and position $(r_\textrm{ref},\theta_\textrm{ref})$. It highlights the limitations of the current form of the distance measure that is, \emph{from its point of view, 2 UEs far away from each other can appear closer due to its periodic and oscillatory nature}. Specifically, for a given reference UE, a UE \emph{A} located on a local minimum of $s^{\star}$ will be considered farther than a UE \emph{B} located on a local maximum even if \emph{A} is closer than \emph{B} in reality (i.e. according to the Euclidean distance between the UEs' locations).

\begin{figure}
	\includegraphics[width=\columnwidth]{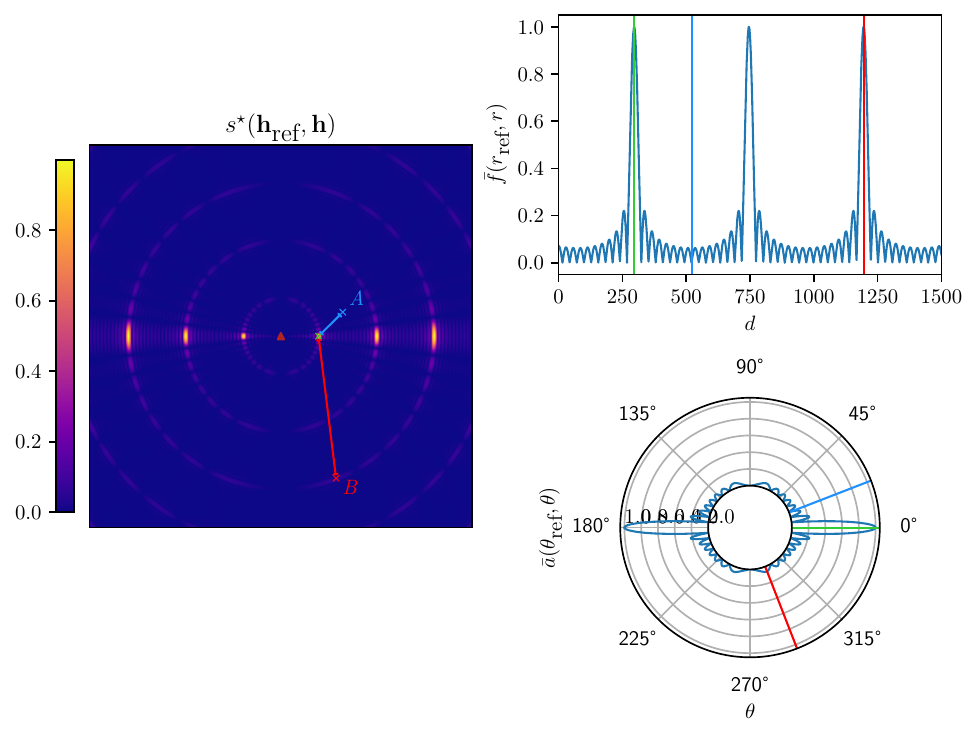}
	\caption{The left figure shows a plot of the similarity $s^\star(\mathbf{h}_\textrm{ref},\mathbf{h})$ between the reference user (green cross and lines) and the rest of the points on the map. The BS is placed at its center. The top right figure shows a plot of $\bar{f}(r_\textrm{ref},r)$ and the bottom right one shows a plot of $\bar{a}(\theta_\textrm{ref},\theta)$. User A (blue cross and lines) is spatially closer to the reference user than user B (red cross/lines), but the distance measure makes the latter appear closer than the former. \new{The objective of this paper is to study theoretically these ambiguities, and propose a new channel distance and system settings that remove them.}}
	\label{fig:func_plot_ULA}
\end{figure}

To deal with the ambiguities mentioned above, system parameters should be chosen guided by the results of Lemma~\ref{lem:lemma_1} and Corollary~\ref{cor:cor_21} with regards to Definition~\ref{def:weak_iden}. 

\noindent\fbox{
\begin{minipage}[t]{0.95\columnwidth}
\begin{proposition}
\label{prop:nec_ULA}
 \textbf{Necessary identifiability condition for ULA.} 
A necessary condition for UEs in an area $A$ to be weakly identifiable by $d^{\star}$, is:
\begin{enumerate}
	\item the radial size $R$ of $A$ should satisfy $R \leq c(\frac{1}{\Delta_f}-\frac{1}{B}) \approx \frac{c}{\Delta_f}$,
	\item for a given $\theta\in[\arcsin(-\frac{1}{N_a}),\arcsin(\frac{1}{N_a})]$, the angular spread of $A$ should remain in the interval $[\arcsin(\sin(\theta)+\frac{N_a-1}{N_a}),\arcsin(\sin(\theta)-\frac{N_a-1}{N_a})]$,
\end{enumerate}
\end{proposition}
\end{minipage}
}
\begin{proof}
Considering that $s^{\star}$ is the product of two terms parametrized by orthogonal parameters, the area of identifiable channels is characterized by deriving ranges on these parameters independently. Essentially, this amounts to determining intervals on $r_i$, $r_j$ and $\theta_i$, $\theta_j$ where only the central main lobes of $\bar{f}$ and $\bar{a}$, respectively, remain.

Since the period of $\bar{f}$ is $D_{\bar{f}}=\frac{c}{\Delta \bar{f}}$, it is straightforward to conclude that for a given $r$ representing the radial center of the area, $r_i$ and $r_j$ should remain in the range $[r-(\frac{D_{\bar{f}}}{2}-\frac{L_{\bar{f}}}{4}),r+(\frac{D_{\bar{f}}}{2}-\frac{L_{\bar{f}}}{4})]$. Half the width of the main lobes (i.e., $\frac{L_{\bar{f}}}{2}$) is subtracted at the extremities to leave them out of the area and avoid resulting ambiguities. This results in an area of radial width of $2\times(\frac{D_{\bar{f}}}{2}-\frac{L_{\bar{f}}}{4})=c(\frac{1}{\Delta_f}-\frac{1}{B})$.

Regarding $\bar{a}$, consider \eqref{eq:second_term} where the period is $D_{\bar{a}}=\frac{1}{\Delta_r}$ w.r.t. $\Theta$. For a given $\theta$ representing the angular center of the area where $sin(\theta)=\Theta$,  $\Theta_i$ and $\Theta_j$ should remain in the range $[\Theta-(\frac{D_{\bar{a}}}{2}-\frac{L_{\bar{a}}}{4}),\Theta+(\frac{D_{\bar{a}}}{2}-\frac{L_{\bar{a}}}{4})]=[\Theta-\frac{N_a-1}{N_a},\Theta+\frac{N_a-1}{N_a}]$ (with $\Delta_r = \frac{1}{2}$). When converted to the angular domain, the angular range is given by $[\arcsin(\sin(\theta)-\frac{N_a-1}{N_a}),\arcsin(\sin(\theta)+\frac{N_a-1}{N_a})]$. To respect the definition domain of $\arcsin$, $\theta$ should remain in the interval $[\arcsin(-\frac{1}{N_a}),\arcsin(\frac{1}{N_a})]$.
\end{proof}

Proposition~\ref{prop:nec_ULA} is illustrated in Fig.~\ref{fig:identifiable_area_ULA}.
Notice how the width of the area is composed of two terms: $\frac{c}{\Delta_f}$ and $\frac{c}{B}$. When the number of subcarriers $N_s$ is high enough, the second one becomes negligible compared to the first, meaning that the radial size of the area is mainly controlled by the subcarrier spacing $\Delta_f$. On the other hand, when the number of antennas is high enough, the angular range of the area tends to $[-\frac{\pi}{2},\frac{\pi}{2}]$, meaning that it covers the whole area in front of the ULA. With this in mind, for UEs in an area of radial width $R$, the subcarrier spacing should be fixed so that $\Delta_f\lessapprox\frac{c}{R}$. Another important remark is that an identifiable area, by definition, should not cross the line of the ULA.

Proposition~\ref{prop:nec_ULA} is necessary because it is required that any channel set identifiable by $d^{\star}$ according to Definition~\ref{def:weak_iden} is included in a well defined region controlled by system parameters. It is defined as the smallest region where a single main lobe appears wherever the reference user is taken, effectively eliminating long-range ambiguities. However, it is not sufficient because short-range ambiguities still subsist.

\begin{figure}
	\includegraphics[width=\columnwidth]{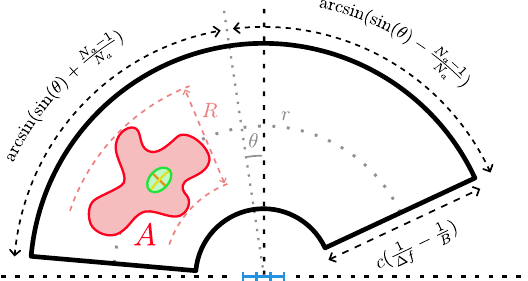}
	\caption{The identifiable area's outline. Any region A of any shape included in there (e.g., the red region) verifies the necessary condition for identifiability. The green patch represents the identifiable neighborhood of the channel at its center. The orange segment is its radial axis and the yellow arc is its angular axis. The blue segment represents the ULA at the location of the BS. $r$ is the radial center of the area and $\theta$ is its angular center.}
	\label{fig:identifiable_area_ULA}
\end{figure}

\noindent\fbox{
\begin{minipage}[t]{0.95\columnwidth}
\begin{proposition}
\label{prop:suff_ULA}
\textbf{Sufficient identifiability condition for ULA.} A sufficient condition for UEs in an area A to be weakly identifiable by $d^{\star}$ is, \emph{in addition to verifying the necessary identifiability condition}, for their similarity value $s^{\star}$ to be above or equal to a threshold $t_s\triangleq\max(\lvert D_{N_s}(3\pi)\rvert,\lvert D_{N_a}(3\pi)\rvert)$.
\end{proposition}
\end{minipage}
}
\begin{proof}
	Given a reference channel $\mathbf{h}_\textrm{ref}$ in the identifiable area, to cut out the secondary lobes of $s^{\star}(\mathbf{h}_\textrm{ref},.)$, only values above their maximum should be retained for the distance measure. Consequently, corresponding thresholds $t_{\bar{f}}$ and $t_{\bar{a}}$ are derived for the frequency and angular terms respectively. The thresholds are deduced from the amplitude of the second highest extremum (because of the absolute value) of the Dirichlet kernel. The first one is given as $t_{\bar{f}} \triangleq \lvert D_{N_s}(3\pi)\rvert=\frac{1}{N_s\sin(\frac{3\pi}{2N_s})}$. Similarly, the second threshold is given as $t_{\bar{a}} \triangleq \lvert D_{N_a}(3\pi)\rvert=\frac{1}{N_a\sin(\frac{3\pi}{2N_a})}$.
	Since $s^\star=\bar{f}\times\bar{a}$ and knowing that $0\leq\bar{f}, \bar{a}\leq 1$, a sufficient condition for both terms to be above their thresholds is to impose $s^{\star}({\mathbf{h}_{\textrm{ref}}},.)\geq t_s \triangleq \max(t_{\bar{f}},t_{\bar{a}})$. 
    The set of channels verifying this condition, called the \emph{identifiable neighborhood} of $\mathbf{h}_{\textrm{ref}}$, is a subset of the set of all identifiable channels.

	For $\forall \vec{x}$ in the identifiable region and $\forall\vec{y},\vec{z}$ in its identifiable neighborhood
	\begin{align}
		&&d^{\star}(\mathbf{h}( \vec{y}),\mathbf{h}( \vec{x})) > d^{\star}(\mathbf{h}( \vec{z}),\mathbf{h}( \vec{x}))\\
		&\Leftrightarrow& \lvert \bar{f}(\mathbf{h}( \vec{y}),\mathbf{h}( \vec{x}))\rvert\times\lvert \bar{a}(\mathbf{h}( \vec{y}),\mathbf{h}( \vec{x}))\rvert>\\
		&&\lvert \bar{f}(\mathbf{h}( \vec{z}),\mathbf{h}( \vec{x}))\rvert\times\lvert \bar{a}(\mathbf{h}( \vec{z}),\mathbf{h}( \vec{x}))\rvert\\
		&\Rightarrow&\lvert \bar{f}(\mathbf{h}( \vec{y}),\mathbf{h}( \vec{x}))\rvert>\lvert \bar{f}(\mathbf{h}( \vec{z}),\mathbf{h}( \vec{x}))\rvert\\
		&&\textrm{or}\ \lvert \bar{a}(\mathbf{h}( \vec{y}),\mathbf{h}( \vec{x}))\rvert>\lvert \bar{a}(\mathbf{h}( \vec{z}),\mathbf{h}( \vec{x}))\rvert
	\end{align}
	Since the identifiable neighborhood corresponds to the main lobes of $\bar{f}$ and $\bar{a}$, this means that they are strictly monotonic in that interval w.r.t. their parameter $r$ and $\theta$ respectively. As a consequence,
	\begin{align*}
		d^{\star}(\mathbf{h}( \vec{y}),\mathbf{h}( \vec{x})) > d^{\star}(\mathbf{h}( \vec{z}),\mathbf{h}( \vec{x}))\\
		\Rightarrow |r_y-r_x| > |r_z-r_x|\ \textrm{or}\ |\theta_y-\theta_x| > |\theta_z-\theta_x|.
	\end{align*}
\end{proof}
Each identifiable channel has its own identifiable neighborhood. In particular, the spatial region associated to it has a shape that resembles an ellipse with a radial linear axis and an angular curved one, corresponding to the radial and angular lobes respectively, as shown on Fig.~\ref{fig:identifiable_area_ULA}. Consequently, any channel which location belongs to that region is in the identifiable neighborhood of the reference channel. This results in a neighborhood graph that can be exploited by other methods to construct the low-dimensional channel chart. A consequence of the thresholding operation involved in the construction of this graph is that the main lobes' widths are changed. The new width of the main lobe of the frequency component, noted $L'_{\bar{f}}$, is the solution to the equation $\left\vert D_{N_s}(\frac{\pi BL'_{\bar{f}}}{c})\right\vert=t_s$. It can be computed through the inversion of $D_{N_s}$ in the interval $[0,2\pi]$ using classical root finding algorithms. 
A similar method can be used to derive the new width of the main lobe of the angular component, noted $L'_{\bar{a}}$ (in rad). It is given this time as the difference between the two roots of the equation $\lvert D_{N_a}(2\pi\Delta_r N_a (\sin\theta_0-\sin\theta))\rvert=t_s$ in the interval $[-\frac{\pi}{2},\frac{\pi}{2}]$. Hence, the radial axis of an identifiable neighborhood is of length $L'_{\bar{f}}$ while its angular axis is of length $L'_{\bar{a}}\times r_\textrm{ref}$.

Although the necessary and sufficient conditions guarantee the weak identifiability of $d^*$ in the case of a ULA, it  can still suffer from the varying geometrical size of the identifiable neighborhoods. Indeed, the width of the main lobe of $\bar{a}(\theta_\textrm{ref},.)$ is not constant and depends on $\theta_
i$. In addition, a ULA is limited to the half-space in front of it because of the symmetry at $\theta= \frac{\pi}{2}\;[\pi]$ (the axis of the ULA) causing ambiguities. The next section dives into how using a UCA can help mitigate these limitations.

\section{UCA are better than ULA for channel charting}
\label{sec:UCA}
In this section, the use of UCAs as a replacement for ULAs is explored. \new{The objective of this section is to show that UCAs are better for charting than ULAs because they do not exhibit axial symmetry and their angular resolution is constant (does not depend on the azimuth), which allows to derive a relation between the bandwidth and the radius of the UCA that limits the warping behavior of CC.}

\begin{lemma}
	\label{lem:a_tilde}
	With a UCA, a good approximation of $\bar{a}$ is of the form
	\begin{equation}
		\tilde{a}(\theta_i,\theta_j)=\left\vert J_0\left(\frac{4\pi}{\lambda}R\left\vert\sin\frac{\theta_i-\theta_j}{2}\right\vert\right)\right\vert,
	\end{equation}
	where $J_0$ is the Bessel integral of order 0. It contains a single main lobe in the interval $[-\pi,\pi]$ centered at $\theta_i=\theta_j$ of width $L_{\tilde{a}}=4\arcsin(\min(\frac{\lambda}{4\pi R}\times j_{0,1},1))$ (in rad) where $j_{0,1}=2.4048$ is the first root of the Bessel integral of order 0~\cite{Weisstein2020}.
\end{lemma}
\begin{proof}
	See Appendix~\ref{app:lemma_3}.
\end{proof}

Based on this approximation, a new distance function is defined,

\begin{equation}
	\tilde{d}^{\star}(\mathbf{h}_i,\mathbf{h}_j)\triangleq\sqrt{2-2\tilde{s}^{\star}(\mathbf{h}_i,\mathbf{h}_j)},
	\label{eq:d}
\end{equation}
where $\tilde{s}^{\star}(\mathbf{h}_i,\mathbf{h}_j)=\bar{f}(r_i,r_j)\times\tilde{a}(\theta_i,\theta_j)$. Fig.~\ref{fig:func_plot_UCA} show a 2D color plot of~\eqref{eq:s} along with plots of $\bar{f}$ and $\tilde{a}$ for some arbitrary system parameters and for a reference user with associated channel $\mathbf{h}_\textrm{ref}$ and location $(r_\textrm{ref},\theta_\textrm{ref})$. Contrary to $s^\star$, this time only one angular main lobe appears.

\begin{figure}
	\includegraphics[width=\columnwidth]{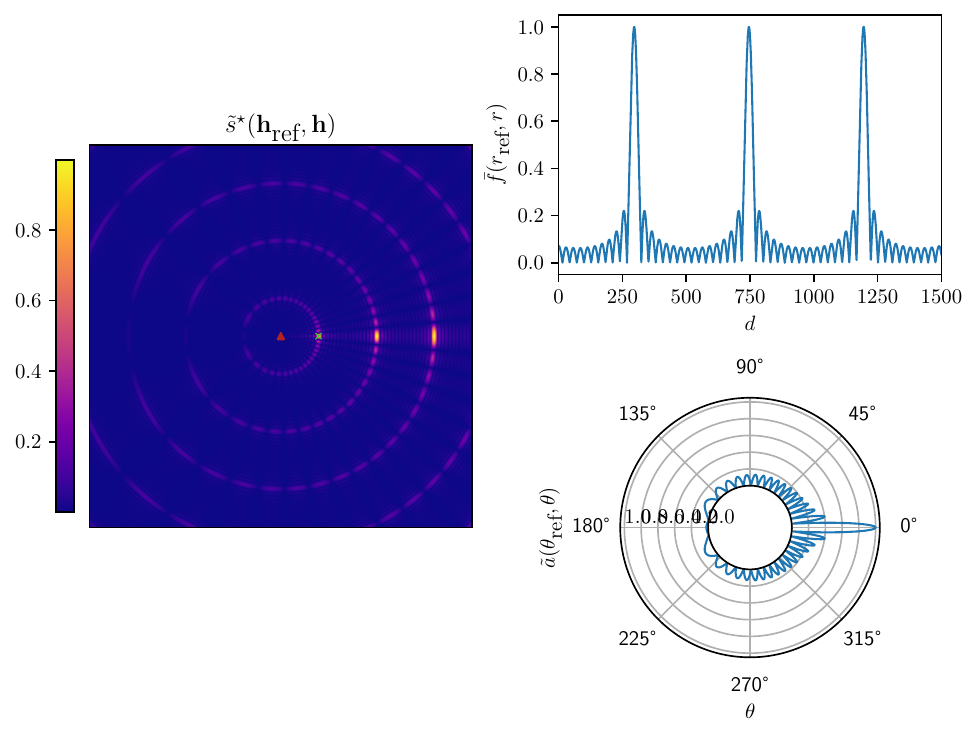}
	\caption{The left figure shows a plot of the similarity $\tilde{s}^\star(\mathbf{h}_\textrm{ref},\mathbf{h})$ between the reference user (green cross) and the rest of the points on the map. The BS is placed at its center. The top right figure shows a plot of $\bar{f}(r_\textrm{ref},r)$ and the bottom right one shows a plot of $\tilde{a}(\theta_\textrm{ref},\theta)$.}
	\label{fig:func_plot_UCA}
\end{figure}

\noindent\fbox{
\begin{minipage}[t]{0.95\columnwidth}
\begin{propositionbis}{prop:nec_ULA}
\label{prop:nec_UCA}
 \textbf{Necessary identifiability condition for UCA.} 
For UEs in an area $A$ to be weakly identifiable by $\tilde{d}^{\star}$, its radial size $R$ should satisfy $R \leq c(\frac{1}{\Delta_f}-\frac{1}{B}) \approx \frac{c}{\Delta_f}$.
\end{propositionbis}
\end{minipage}
}

\begin{proof}
	The first part of the proof is the same as the one involving $\bar{f}$ in the proof of Proposition~\ref{prop:nec_ULA}.

	Regarding $\tilde{a}$, the function contains a single main lobe in its domain of definition according to Lem.~\ref{lem:a_tilde} and is thus not subject to long-range ambiguities that can be caused by the period nature of a a function.
\end{proof}

Proposition~\ref{prop:nec_UCA} is illustrated Fig.~\ref{fig:identifiable_area_UCA}. It highlights the fact that with a UCA, the identifiable area is no longer angularly restricted and is extended to include the whole angular domain (i.e., $[-\pi,\pi]$) with no ambiguities.

\noindent\fbox{
\begin{minipage}[t]{0.95\columnwidth}
\begin{propositionbis}{prop:suff_ULA}
\label{prop:suff_UCA}
\textbf{Sufficient identifiability condition for UCA.} A sufficient condition for UEs in an area A to be weakly identifiable by $\tilde{d}^{\star}$ is, \emph{in addition to verifying the necessary identifiability condition and for $N_s>2$}, for their similarity value $\tilde{s}^{\star}$ to be above or equal to a threshold $t_{\tilde{s}}\triangleq 0.403$.
\end{propositionbis}
\end{minipage}
}

\begin{proof}
	The proof follows the same logic as that of Proposition~\ref{prop:suff_ULA} with the difference that this time the threshold of the angular term is changed and is derived from the amplitude of the second highest extremum of the Bessel integral. It is straightforwardly given as $t_{\tilde{a}}\triangleq\lvert J_0(j'_{0,1})\rvert=0.403$ where $j'_{0,1}=3.8317$ is the first root of the derivative of the Bessel integral of order 0~\cite{Weisstein2020}. The global threshold is given as $t_{\tilde{s}} \triangleq \max(t_{\bar{f}},t_{\tilde{a}})$. Finally, knowing that $t_{\bar{f}}$ is strictly decreasing with $N_s$, and that $t_{\bar{f}}<t_{\tilde{a}}$ for $N_s>2$, 
    it follows that $t_{\tilde{s}}=t_{\tilde{a}}$ for $N_s>2$.
\end{proof}

This time again, thresholding has an impact on the widths of the main lobes. Since $t_{\tilde{s}}$ has now a constant value, a closed-form expression of $L'_{\bar{f}}$ can be derived through inversion of $\left\vert D_{N_s}(\frac{\pi BL'_{\bar{f}}}{c})\right\vert=t_{\tilde{s}}$. Although the root depends on $N_s$, its actual value varies negligibly with it and is thus derived for $N_s\rightarrow\infty$ as an acceptable approximation. As a result, the width of the radial main lobe after thresholding is given by $L'_{\bar{f}}=\frac{c\times 4.238}{\pi B}$. Similarly, the new width of $\tilde{a}$ (in rad), noted $L'_{\tilde{a}}$, is given as the solution to the equation $\left\vert J_0(\frac{4\pi}{\lambda}R\left\vert\sin\frac{L'_{\tilde{a}}}{4}\right\vert)\right\vert=t_{\tilde{s}}$. It can be computed through the inversion of $J_0$ in the interval $[0,j_{0,1}]$. It's value is given by $L'_{\bar{a}}=4\arcsin(\frac{\lambda}{4\pi R}\times 1.692)$.

\begin{figure}[t!]
\centering
	\includegraphics[width=0.8\columnwidth]{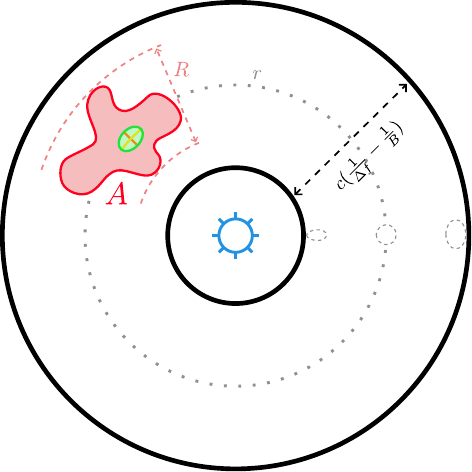}
	\caption{The identifiable area's outline. Any region A of any shape included in there (e.g., the red region) verifies the necessary condition for identifiability. The green patch represents the identifiable neighborhood of the channel at its center. The orange segment is its radial axis and the yellow arc is its angular axis. The blue circle represents the UCA at the location of the BS. $r$ is the radial center of the area.}
	\label{fig:identifiable_area_UCA}
\end{figure}

\noindent\fbox{
\begin{minipage}[t]{0.95\columnwidth}
\begin{proposition}
\label{prop:UCA}
\textbf{Constant angular spread.} When using the thresholded channel distance, the UCA (as opposed to the ULA) yields neighborhoods whose angular spread do not depend on the azimuth.
\end{proposition}
\end{minipage}
}
\begin{proof}
	The proof directly follow from the fact that $L'_{\bar{a}}$, which represents the angular spread of the obtained neighborhoods, doesn't depend on $\theta$.
\end{proof}
This is a very interesting property since it implies that the angular distance perceived by the system is independent of the azimuth, thus avoiding warping in the obtained channel chart, as illustrated in Sec.~\ref{ssec:ula_vs_uca}. However, the size of the resulting identifiable neighborhoods are still not constant everywhere. To ensure that the distance reflects the true geometrical neighborhoods without distortion, the identifiable neighborhoods should be as round as possible. This is measured by the ratio of both their axis, i.e. $\gamma\triangleq\frac{L'_{\bar{f}}}{L'_{\tilde{a}}\times r_\textrm{ref}}\approx\frac{c}{Br_\textrm{ref}}\frac{1.06}{\pi \arcsin\left( \frac{\lambda}{4\pi R}.1.692 \right)}$.

\noindent\fbox{
\begin{minipage}[t]{0.95\columnwidth}
\begin{proposition}
\label{prop:round}
\textbf{Round neighborhoods.} For a BS equipped with a UCA of radius $R$, a bandwidth $B$ and an area considered for charting whose center is at a distance $r_0$ from the BS, the neighborhoods identified by the thresholded PI distance are round at the center of the area if and only if $1=\gamma\approx \frac{1.06c}{Br_0\pi \arcsin\left( \frac{\lambda}{4\pi R}.1.692 \right)}$.
\end{proposition}
\end{minipage}
}
\begin{proof}
	Since an identifiable neighborhood has the shape of an ellipse approximately, it is round when its axes' lengths, $L'_{\bar{f}}$ and $L'_{\tilde{a}}\times r_0$, are equal. As a consequence, $\gamma=1$.
\end{proof}

\begin{figure}[t!]
	\includegraphics[width=\columnwidth]{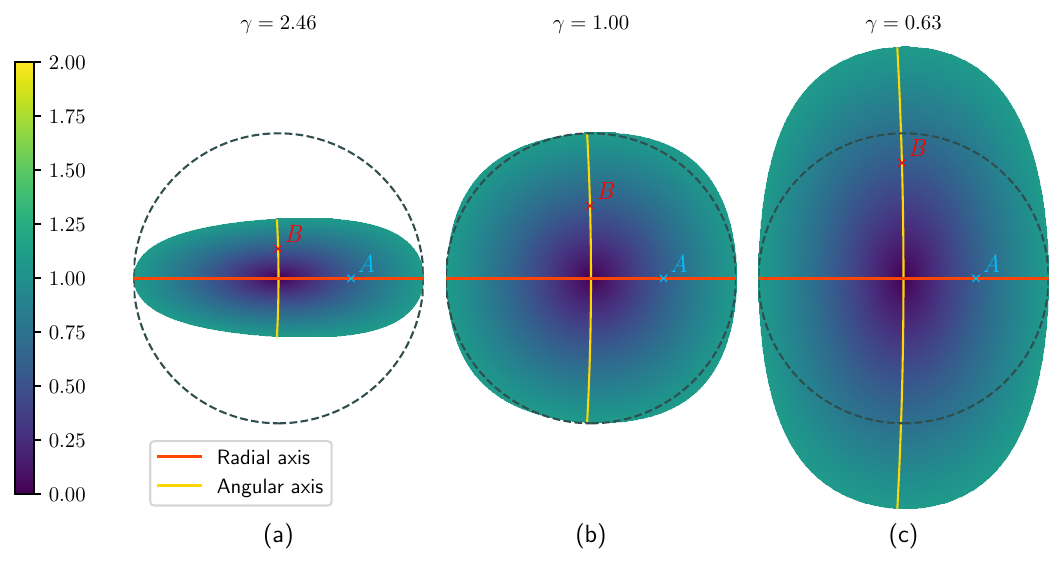}
	\caption{Visualization of the distance measure for identifiable neighborhoods at (a) the inner, (b) center and (c) outer edge of the identifiable area as pictured in Fig.~\ref{fig:identifiable_area_UCA}. The circle outline helps visualize how their shapes deviate from being perfectly round. UEs A and B are considered at the same distance from the point of view of the PI distance.}
	\label{fig:neighborhoods}
\end{figure}

This result allows to determine the optimal bandwidth $B$ as a function of the distance $r_0$ and the radius $R$ as
$B_{\text{opt}} = \frac{c}{r_0}\frac{1.06}{\pi \arcsin\left( \frac{\lambda}{4\pi R}.1.692 \right)}$. It also allows to determine the optimal radius $R$ as a function of the distance $r_0$ and the bandwidth $B$ as $R_{\text{opt}} = \frac{\lambda . 1.692}{4\pi\sin\left( 
\frac{c}{Br_0} \frac{1.06}{\pi} \right)}$.

Fig.~\ref{fig:neighborhoods} shows the shape of the identifiable neighborhoods when the reference user is at the inner edge of the identifiable area, at it's center and at it's outer edge, from left to right respectively. Because of distortion, users A and B will be considered at the same distance from the reference user as measured by \eqref{eq:s}, although it's only true for the center identifiable neighborhood. Furthermore, its area can be approximated with the area of the circle of diameter $L'_f$ as $\frac{\pi L_f^{'2}}{4}$, as shown on the figure. Since for better performance a minimum number of neighboring users $k_{\textrm{min}}$ should be located inside the area defined by the patch for each reference user, the minimum density of users is given as $\frac{4k_{\textrm{min}}}{\pi L_{\bar{f}}^{'2}}$. In reality, this density should be a bit higher to account for the distortion of the neighborhood going towards the inner part of the identifiable area. In practice, the number of neighbors inside of the identifiable neighborhoods will vary for each reference user, but there should be at least one to connect it to the neighborhood graph constructed by Isomap.

\neww{
\section{Discussion}
\label{sec:disc}
\subsection{Practical consequences of the theoretical findings}
The theoretical results given in the previous sections can be leveraged in practice to optimize channel charting results. Indeed, each proposition can be used to provide concrete guidelines on the system settings. As an example, consider  a scenario in which channel charting has to be carried out with users located in an area $A$ of radial size $R$ and center $r_0$ (as in Fig.~\ref{fig:identifiable_area_UCA}). Then, the following can be stated: 
\begin{itemize}
	\item The system should have a subcarrier spacing $\Delta_f \leq \frac{c}{R}$ (Prop.~\ref{prop:nec_ULA} and Prop.~\ref{prop:nec_UCA}) in order to avoid radial ambiguities.
	\item The BS should be equipped with a UCA of a given radius $R_\textrm{UCA}$ in order to avoid angular ambiguities (Prop.~\ref{prop:nec_UCA}) and ensure a constant angular resolution (Prop.~\ref{prop:UCA}).
	\item The bandwidth $B$ and UCA radius $R_\textrm{UCA}$ should be chosen so as to satisfy $\frac{1.06c}{Br_0\pi \arcsin\left( \frac{\lambda}{4\pi R_\textrm{UCA}}.1.692 \right)}=1$ to ensure having the roundest neighborhoods possible at the radial center $r_0$ of the area to chart (Prop.~\ref{prop:round}).
	\item The computed PI distance should be thresholded in order to keep only values below $\sqrt{2-2\times t_{\tilde{s}}} = 1.093$ to avoid oscillations due to secondary lobes (Prop.~\ref{prop:suff_UCA}).
\end{itemize}
Following all these guidelines strongly improves channel charting results, as shown empirically in the following section.}


\neww{\subsection{Generalization to multipath scenarios}
The analysis conducted in this paper is limited to single path LoS channels. However, it can be directly extended to handle multipath propagation, provided a LoS path exists. Indeed, once channels are estimated, one can envision extracting the LoS component using a sparse recovery algorithm. For example, this can be done with a single iteration of a greedy algorithm such as matching pursuit \cite{Mallat1993} and a dictionary whose atoms are products of steering vectors and frequency response vectors. Then, the proposed analysis and charting strategy can be applied as is on LoS components extracted from multipath channels.
}

\neww{\subsection{Importance of identifiability for applications}
Having an identifiable distance measure with consistent neighborhoods and no ambiguities translates into learning channel charts of great quality. Such charts constitute a solid base for subsequent tasks that heavily rely on spatial features like radio resource management. For example, pilot allocation usually relies on reusing non-orthogonal pilot sequences amongst nonadjacent UEs. CC can be used in this context to avoid pilot contamination by providing charts reflecting this nonadjacency~\cite{Ribeiro2020}. It is thus important to base the charting process on a distance free of ambiguities. Another example is that of CC-based beamforming where the same principle applies so that beams are allocated based on trustworthy and consistent spatial information~\cite{Yassine2022b}. Conversely, it the ambiguities are not effectively eliminated, the risk is for the learned charts to convey inconsistent indications that could degrade the performance of downstream tasks. 
}

\section{Experiments}
\label{sec:exp}
In this section, the theoretical claims and system setting guidelines of Sec.~\ref{sec:ULA}, Sec.~\ref{sec:UCA} \neww{and Sec.~\ref{sec:disc}} are empirically validated both on synthetic and realistic channel data. \new{The Isomap dimensionality reduction algorithm \cite{Tenenbaum2000} is used to produce the final chart, unless otherwise stated. Key parameters are summarized in Table~\ref{tab:parameters}.}


\begin{table}[h]
	\neww{
	\centering
	\begin{tabular*}{\columnwidth}{@{\extracolsep{\fill}}rccc@{\extracolsep{\fill}}}
	\toprule
	\multicolumn{1}{c}{\multirow{2}{*}{Parameter}} & \multicolumn{3}{c}{Experiment}                          \\ \cmidrule(l){2-4} 
	\multicolumn{1}{c}{}                           & Sec.~\ref{ssec:toy} & Sec.~\ref{ssec:ula_vs_uca} & Sec.~\ref{ssec:real}     \\ \midrule
	Bandwidth                                      & 10 MHz & 10 MHz  & 7 MHz \\
	Number of subcarriers                          & 16 & 16      & 19      \\
	Carrier frequency                              & 3 GHz & 3 GHz & 3 GHz                               \\
	Type of array            & ULA          & UCA          & UCA   \\ 
 Number of antennas & 16           & 64           &  64    \\
	Number of samples                                  & 2838 & 2838  & 3383  \\ \bottomrule
	\end{tabular*}
	\caption{\neww{Summary of key simulation parameters for experiments of Sec.~\ref{sec:exp}}}
	\label{tab:parameters}
	}
\end{table}

\subsection{Practical influence of system parameters}
\label{ssec:toy}

A first simple experiment is considered, in which the effect of system parameters on the obtained channel chart is evidenced, showing the practical importance of the theoretical claims of the previous sections. The base scenario consists of a BS equipped with a UCA of 64 antennas and of radius $R=42~\textrm{cm}$ operating at $f_c=3~\textrm{GHz}$. It transmits over a bandwidth of $B=10~\textrm{MHz}$ and $N_s=16$ subcarriers. A total of $N_c=2838$ channel observations of UEs located in front of the BS is generated according to~\eqref{eq:channel_model}.
These parameters are used to compute and characterize the identifiable area, the identifiable neighborhoods as well as the optimal distance for the radial center $r_0$ according to the analysis of the previous sections. The identifiable area's width is of $422~\textrm{m}$ (Proposition~\ref{prop:nec_UCA}) and its radial center should be placed at a distance of $296~\textrm{m}$ from the BS to limit neighborhood distortion (Proposition~\ref{prop:round}). Fig.~\ref{fig:toy_example_threshold} shows a visualization of the distance between the UE at the center of the area (i.e., $r_\textrm{ref}=r_0$) and the rest of the UEs before and after thresholding. The base scenario is tested against 3 variants:
\begin{enumerate}
	\item not applying the threshold;
	\item halved width of the identifiable area by doubling the subcarrier spacing $\Delta _f$ and keeping the same bandwidth so that some users end up outside of it;
	\item non-round neighborhoods at the center by reducing the bandwidth $B$.
\end{enumerate}

Fig.~\ref{fig:toy_example_chart} shows the learned charts for the base scenario as well as for each variant. The word "b$<>$com" 
helps visualize how the global structure of the locations is preserved on it. The chart obtained with the first variant looks disorganized. Although long-range ambiguities are eliminated, short-range ones remain as no threshold is applied, causing a degradation in performance in particular in terms of TW and KS. The chart obtained with the second variant is warped as the parts of the considered area that are outside of the identifiable area are confused with other parts inside of it because of long-range ambiguities. As a consequence, TW and KS are negatively impacted. This shows the importance of limiting UEs to the identifiable zone. The chart obtained with the last variant is flattened as a consequence of the non-round identifiable neighborhoods, meaning that distances are perceived differently on the radial and angular directions. While there is little impact on TW and CT, KS is highly impacted because the global structure of the chart is not preserved. Finally, the chart obtained in the base scenario has the best overall performance and, although not perfect in terms of the global structure, represents a pretty faithful representation of the locations of the UEs. 

\begin{figure}
	\includegraphics[width=\columnwidth]{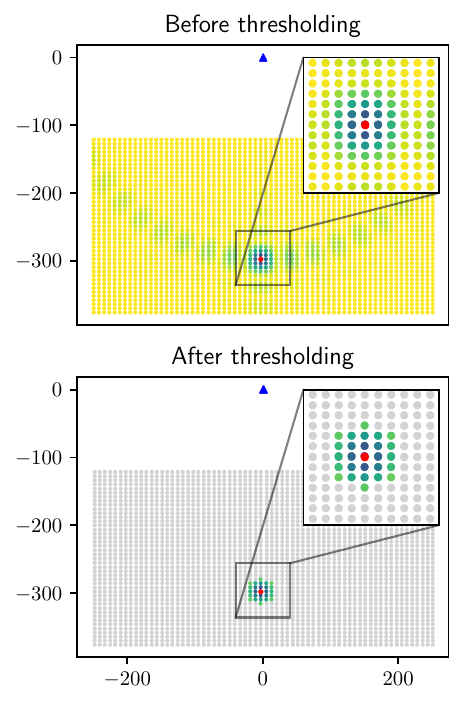}
	\caption{The computed PI distance between a reference UE (red point) and the rest of the UEs before (top) and after (bottom) thresholding. The greyed points represent the points that are outside the identifiable neighborhood and thus discarded.}
	\label{fig:toy_example_threshold}
\end{figure}

\begin{figure}
	\includegraphics[width=\columnwidth]{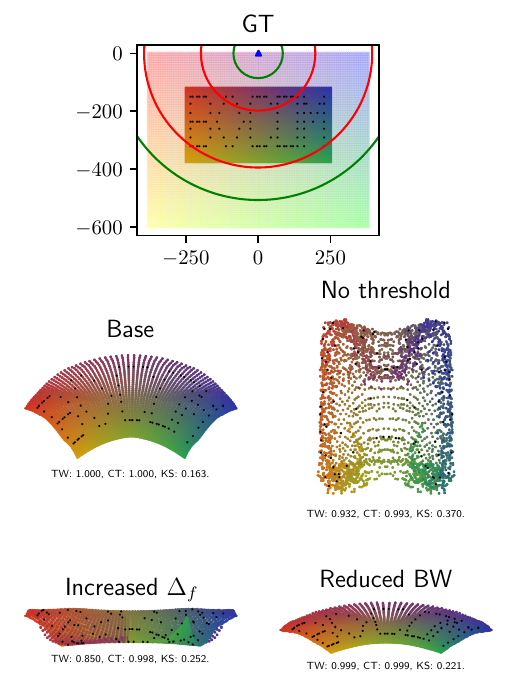}
	\caption{The GT UEs' locations and the learned charts corresponding to the base scenario and its three variants.}
	\label{fig:toy_example_chart}
\end{figure}

\subsection{ULA vs. UCA}
\label{ssec:ula_vs_uca}
Another experiment is considered, this time to further empirically assess the advantage of a UCA over a ULA, in order to validate the theoretical claims of Sec.~\ref{sec:UCA}. A BS placed at the center of a map communicates with UEs uniformly distributed around it. The same system parameters as in the base scenario of the previous section are considered, with the difference that in one case a ULA of 16 antennas is considered, while in the other case a UCA of 64 antennas is considered. The difference in the number of antennas is to obtain lobes of comparable radial sizes in both cases. In addition, two areas  that include the UEs' positions are considered: one that is identifiable according to the UCA and one that is identifiable according to the ULA also. The second one is included in the first one by definition. Note that the communication aspect of the system is ignored. Fig.~\ref{fig:ULA_vs_UCA} shows the ground truth UEs' locations as well as the learned charts in each case. In the case of the red area, only points going from yellow to pink appear on the chart associated with the ULA, which represent locations on the right of the map. This is because they overlap with the locations on its left. The channel distance considers that points on the right coincide with points on the right due to its axial symmetry w.r.t. the vertical axis on the figure. Isomap then proceed to learn a chart that places the wrongly matched points at the same locations. On the other hand, the chart produced in the case of the UCA successfully reflects the circular structure of the locations. The TW score in particular being low in the case of the ULA indicates that a lot of false neighborhoods are introduced in the chart which is not the case in the case of the UCA, confirming the observations. In the case of the orange area, both charts appear to be of good quality. However, upon closer inspection, it appears that the chart associated with the UCA is of better quality as indicated by the lower (i.e., better) KS score. Indeed, the points appear to be distributed uniformly whereas in the case the ULA, their density seems to increase the more pink they get. This can be explained by the angular width of the identifiable neighborhoods being dependent on the azimuth in the case of the ULA.  

\begin{figure}
	\includegraphics[width=\columnwidth]{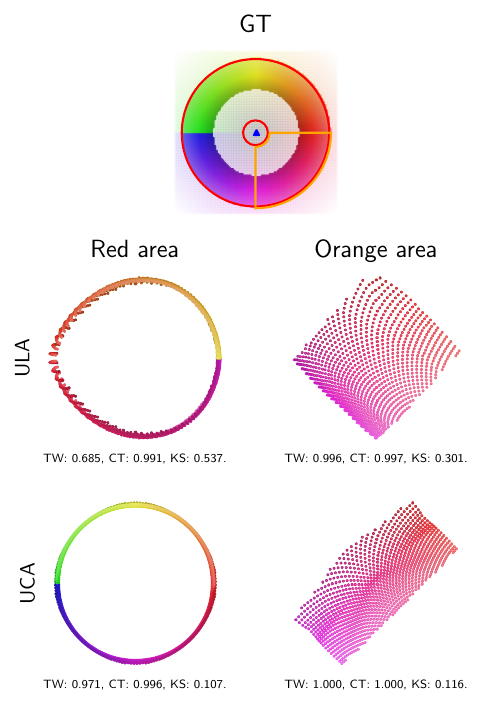}
	\caption{The GT UEs' locations and the learned charts corresponding to the use of a ULA vs. a UCA for two different areas (red and orange).}
	\label{fig:ULA_vs_UCA}
\end{figure}


\subsection{Realistic scenario}
\label{ssec:real}
\begin{figure*}	
	\includegraphics[width=\textwidth]{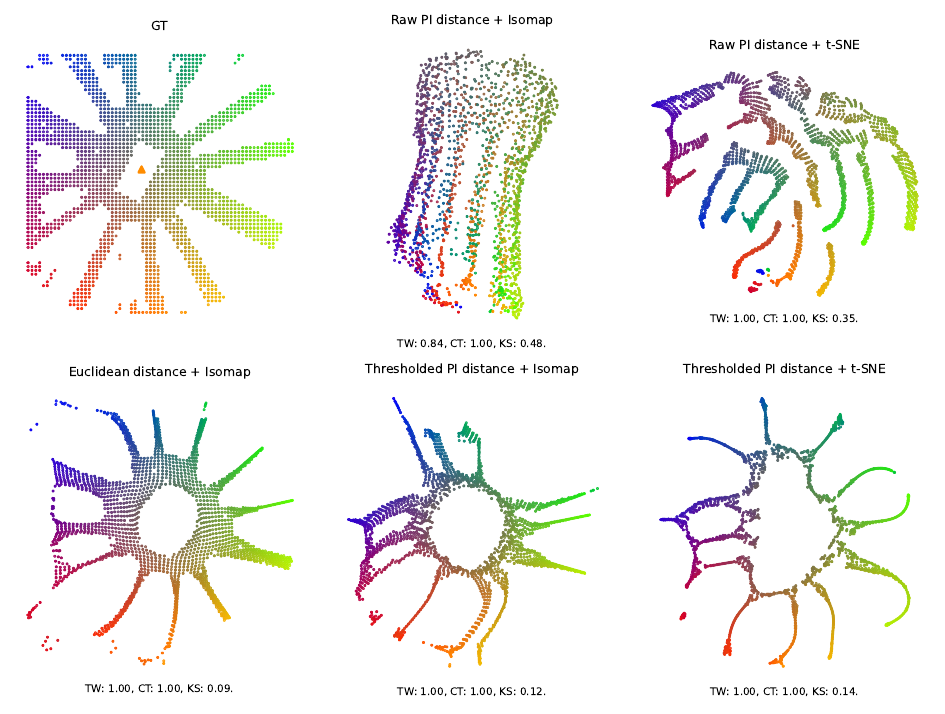}
	\caption{The GT UEs' locations for the realistic scenario and the learned charts corresponding to the use of the Euclidean distance (gold standard), the raw PI distance and the \new{proposed thresholded PI distance for both Isomap and t-SNE.}}
	\label{fig:sionna}
\end{figure*}
In this experiment, the presented CC approach is tested on a more realistic scenario using the Sionna ray tracing engine~\cite{Hoydis2023} on the ``etoile" scene, which is a 3D reconstruction of the area around the Arc De Triomphe in Paris. The BS is placed at the center of the scene and above the arc. It is equipped with a UCA of 64 elements and a radius of 20 cm. It communicates at a frequency of 3 GHz and over a bandwidth of 7 MHz. A total of $N_c=3383$ UEs are distributed around it and chosen so that they meet the necessary identifiability condition of Proposition~\ref{prop:nec_UCA}, effectively avoiding long-range ambiguities, and only the LoS path is considered. Note that although the theory developed in this paper considers the BS and UEs to be on the same plane, this scenario however considers the more realistic case where the BS is at a height of 70 m while the UEs are 1.5 m above the ground. Fig.~\ref{fig:sionna} shows the ground truth (GT) locations and the charts generated by \new{two dimensionality reduction methods, namely Isomap \cite{Tenenbaum2000} and t-SNE \cite{Maaten2008}} using three different channel distances. First, the Euclidean distance is used to showcase the capabilities of Isomap with quasiperfect knowledge and to serve as a gold standard. In particular, a neighborhood graph is constructed where neighbors are limited to points present in a radius $\frac{L'_{\bar{f}}}{2}$ of each node. This amounts to restricting the knowledge to identifiable neighborhoods according to $\tilde{d}^{\star}$ as defined in Proposition~\ref{prop:suff_UCA}. The resulting chart represents what can be achieved with a strongly identifiable distance (see Definition~\ref{def:strong_iden}), had one been available. The next chart is generated by constructing the neighborhood graph based on the raw PI channel distance, including its short-range ambiguities (i.e., not applying the threshold so that the sufficient identifiability condition of Proposition~\ref{prop:suff_UCA} is not met). The locations on the chart look warped and spatial neighborhoods are hardly distinguishable one from another when using Isomap, \new{ while t-SNE seems to behave better, although with neighborhoods that seem to be randomly layed out.} The final charts are similarly generated by constructing the neighborhood graph based on the channel distance but this time applying the threshold (meeting the sufficient condition of Proposition~\ref{prop:suff_UCA}) and effectively restricting neighborhoods to identifiable ones, which has the effect of eliminating the short-range ambiguities. This results in a much better chart that closely resemble the GT. As expected, thresholding allows to greatly improve TW and KS of the obtained chart, which is perfectly in line with the results of Sec.~\ref{ssec:toy}. \new{Note that for a fair comparison, t-SNE relies on the same geodesic distance as Isomap, obtained from a graph built from the thresholded PI distance. In this case however Isomap seems to slightly outperform t-SNE.} \new{Moreover,it is to be stressed that this paper proposes a new channel distance and not a new charting algorithm (the proposed distance can be used within any distance-based channel charting method), and we use here two charting methods to show that the improvement brought by the proposed distance are not specific to a given charting algorithm.}


\section{Conclusion and perspectives}
In this paper, a theoretical framework for the identifiability of UEs from their channel measurements is introduced in the perspective of adequately setting the parameters of a MIMO OFDM system for channel charting purpose. To that end, a PI channel distance is analyzed, highlighting its short and long-range ambiguities. A practical method for mitigating these limitations and eliminating ambiguities is then proposed by translating the introduced identifiability rules into \new{both guidelines for} the design of the MIMO OFDM systems \new{ and an appropriate threshold on the PI channel distance}. The \new{proposed thresholded PI distance has been used for CC with} a MIMO OFDM system and the conducted experiments show a great improvement in the UEs' neighborhood consistency as long as the established designing rules are applied. In addition, the advantage of using UCAs instead of ULAs for CC is demonstrated and validated through simulations. All the achieved results show a very good concordance with the theory and make the proposed system design method very promising.

\new{In future work, the proposed method should be extended to more realistic scenarios including multipath channels, edge UEs, etc., and it should be tested on real-world measured data.} Finally, it should be considered and extended in the context of ISAC systems with the goal of achieving a trade-off between communication and sensing.


\bibliographystyle{IEEEtran}
\bibliography{biblio.bib}

\appendices

\section{Proof of Lemma~\ref{lem:lemma_1}}
\label{app:lemma_1}
The frequency term can be expressed as
\begin{align}
	\begin{split}
		\bar{f}(r_i,r_j) & =\left\vert \left(\mathbf{f}(r_i)^H\mathbf{f}(r_j)\right)\right\vert                                               \\
		& =\frac{1}{N_s}\left\vert \sum^{N_s}_{s=1} e^{-j2\pi\frac{r_i-r_j}{c}(f_s-f_c)}\right\vert                                        \\
		& \overset{(a)}{=}\frac{1}{N_s}\left\vert \sum^{N_s}_{s=1} e^{-j2\pi\frac{r_i-r_j}{c}\Delta_f(s-1)}e^{-j2\pi\frac{r_i-r_j}{c}(f_c-\Delta_f\frac{N_s-1}{2})}\right\vert                                        \\
		& =\frac{1}{N_s}\left\vert\sum^{N_s-1}_{s=0} e^{-j2\pi\frac{r_i-r_j}{c}\Delta_fs}\right\vert \\
		& \overset{(b)}{=}\frac{1}{N_s}\left\vert\frac{\sin(\pi \frac{B(r_i-r_j)}{c})}{\sin(\pi\frac{B(r_i-r_j)}{cN_s})}\right\vert\\
		& =\left\vert D_{N_s}\left(\frac{2\pi B(r_i-r_j)}{c}\right)\right\vert.
	\end{split}
\end{align}
where $(a)$ is obtained since $f_s=f_c-\Delta_f\frac{N_s-1}{2}+\Delta_f(s-1)$, and $(b)$ comes from the computation of the geometric series sum. 


\section{Proof of Lemma~\ref{lem:lemma_2}}
\label{app:lemma_2}

The angular term can be expressed as

\begin{align}
	\begin{split}
		\bar{a}(\theta_i,\theta_j) & =\left\vert\mathbf{a}(\theta_i)^H\mathbf{a}(\theta_j)\right\vert\\
		&=\frac{1}{N_a}\left\vert\sum_{n=0}^{N_a-1}e^{-j\frac{2\pi}{\lambda}\vec{p}_{n}(\vec{u}(\theta_i)-\vec{u}(\theta_j))}\right\vert.\\
	\end{split}
\end{align}

In particular, for a ULA and using the adequate steering vectors~\eqref{eq:steer_vec}

\begin{align}
	\begin{split}
		\bar{a}(\theta_i,\theta_j) & =\frac{1}{N_a}\left\vert\sum_{n=0}^{N_a-1}e^{-j2\pi\Delta_r n(\sin\theta_i-\sin\theta_j)}\right\vert \\
		&=\frac{1}{N_a}\left\vert\frac{\sin(\pi\Delta_r N_a(\sin\theta_i-\sin\theta_j))}{\sin(\pi\Delta_r(\sin\theta_i-\sin\theta_j))}\right\vert\\
		&=\left\vert D_{N_a}(2\pi\Delta_r N_a (\sin\theta_i-\sin\theta_j))\right\vert.
	\end{split}
\end{align}

\section{Proof of Lemma~\ref{lem:a_tilde}}
\label{app:lemma_3}
Considering a UCA and using the adequate steering vectors~\eqref{eq:steer_vec}, $\bar{a}(\theta_i,\theta_j)$ becomes
\begin{align}
	\begin{split}
		\bar{a}(\theta_i,\theta_j) & =\vert\mathbf{a}(\theta_i)^H\mathbf{a}(\theta_j)\vert                                                                                          \\
		& =\frac{1}{N_a}\left\vert\sum_{i=1}^{N_a}\exp\left(-j\frac{2\pi}{\lambda}R\right.\right.\\
  &\Big(\cos(2\pi\frac{i}{N_a})(\cos\theta_i-\cos\theta_j)\\
  &\left.\left.+\sin(2\pi\frac{i}{N_a})(\sin\theta_i-\sin\theta_j)\Big)\right)\right\vert,
	\end{split}
	\label{eq:a_bar_discrete}
\end{align}
for the discrete array case. Considering the continuous-aperture case by setting $N_a\rightarrow \infty$ leads to
\begin{align*}
	\tilde{a}(\theta_i,\theta_j) & =\frac{1}{2\pi}\left\vert\int_0^{2\pi}\exp\bigg(-j\frac{2\pi}{\lambda}R\Big(\cos\gamma(\cos\theta_i-\cos\theta_j)\right.\\
 &\left.+\sin\gamma(\sin\theta_i-\sin\theta_j)\Big)\bigg)\,d\gamma\right\vert                                                                                                                                   \\
	                               & =\frac{1}{2\pi}\bigg\vert\int_0^{2\pi}\exp\bigg(-j\frac{2\pi}{\lambda}R\\
								   &\Big(\cos(\gamma-\theta_i)-\cos(\gamma-\theta_j)\Big)\bigg)\,d\gamma\bigg\vert.                                                                                                                                                                \\
	\intertext{Substituting $\omega$ for $\gamma-\theta_j$, and with $d\omega=d\gamma$ we have}
	                               & =\frac{1}{2\pi}\bigg\vert\int_{-\theta_j}^{2\pi-\theta_j}\exp\bigg(-j\frac{2\pi}{\lambda}R\\
								   &\Big(\cos(\omega+\theta_j-\theta_i)-\cos(\omega)\Big)\bigg)\,d\omega\bigg\vert,                                                                                                                                                 \\
	\intertext{with $\Delta\theta=\theta_j-\theta_i$. Since $\cos$ is $2\pi$ periodic, we have}
	                               & =\frac{1}{2\pi}\bigg\vert\int_{0}^{2\pi}\exp\bigg(-j\frac{2\pi}{\lambda}R\\
                                &\Big(\cos(\omega+\theta_j-\theta_i)-\cos(\omega)\Big)\bigg)\,d\omega\bigg\vert.                                                                                                                                                              \\
	\intertext{With the help of basic trigonometric identities, it can be shown that $\cos(\omega+\theta_j-\theta_i)-\cos(\omega)=2\left\vert\sin\frac{\theta_j-\theta_i}{2}\right\vert\cos(\omega + \mathrm{atan2}(\sin(\theta_j-\theta_i),\cos(\theta_j-\theta_i)-1))$, 
 leading to}
	                               & =\frac{1}{2\pi}\bigg\vert\int_{0}^{2\pi}\exp\bigg(-j\frac{4\pi}{\lambda}R\left\vert\sin\frac{\theta_j-\theta_i}{2}\right\vert\\
                                &\cos\Big(\omega + \mathrm{atan2}(\sin(\theta_j-\theta_i),\cos(\theta_j-\theta_i)-1)\Big)\bigg)\,d\omega\bigg\vert.                                                                                    \\
	\intertext{Subsituting $\gamma$ for $\omega + \mathrm{atan2}(\sin(\theta_j-\theta_i),\cos(\theta_j-\theta_i)-1) + \frac{\pi}{2}$, and with $d\gamma=d\omega$ we have}
	                               & =\frac{1}{2\pi}\bigg\vert\int_{\mathrm{atan2}(\sin(\theta_j-\theta_i),\cos(\theta_j-\theta_i)-1) + \frac{\pi}{2}}^{2\pi+\mathrm{atan2}(\sin(\theta_j-\theta_i),\cos(\theta_j-\theta_i)-1) + \frac{\pi}{2}}\\
                                &e^{-j\frac{4\pi}{\lambda}R\left\vert\sin\frac{\theta_j-\theta_i}{2}\right\vert\sin\gamma }\,d\gamma\bigg\vert. \\
	\intertext{Finally, exploiting once again the $2\pi$-periodic property of the sine, we have}
	                               & =\frac{1}{2\pi}\left\vert\int_{0}^{2\pi}e^{-j\frac{4\pi}{\lambda}R\left\vert\sin\frac{\theta_j-\theta_i}{2}\right\vert\sin\gamma }\,d\gamma\right\vert.                                                                                                                                               \\
\end{align*}
\balance
This leads to a convenient way of expressing $\tilde{a}$ as it can be written as
\begin{equation}
	\tilde{a}(\theta_i,\theta_j)=\left\vert J_0\left(\frac{4\pi}{\lambda}R\left\vert\sin\frac{\theta_j-\theta_i}{2}\right\vert\right)\right\vert.
\end{equation}

The validity of the integral form as an approximation was empirically verified and holds well when a great number of antennas $N_a$ is considered.



\end{document}